\documentclass[11pt,letterpaper]{article}

\usepackage[T1]{fontenc}
\usepackage{graphicx}
\usepackage[utf8]{inputenc}
\usepackage{cite}
\usepackage{amsmath,amssymb,amsfonts}
\usepackage{ascmac}
\usepackage{amsthm}
\usepackage{textcomp}
\usepackage{xcolor}
\usepackage{bm}
\usepackage{cite}
\usepackage[linesnumbered, ruled, noline, noend]{algorithm2e}
\usepackage{setspace}
\usepackage{comment}
\usepackage{footnote}
\usepackage{color}
\usepackage{authblk}
\usepackage{thmtools, thm-restate}
\usepackage[margin=3.5cm]{geometry}

\newcommand{\Vint}{V_{intm}}
\newcommand{\Sint}{S_{intm}}
\newcommand{\Gtilde}{\tilde{G}}
\newcommand{\Vtilde}{V(\tilde{G})}
\newcommand{\Etilde}{E(\tilde{G})}
\newcommand{\Mtld}{\tilde{M}}

\newtheorem{proposition}{Proposition}
\newtheorem{lemma}{Lemma}

\newtheorem{theorem}{Theorem}

\title{Meeting Times of Non-atomic Random Walks}

\author{Ryota Eguchi \footnote{Nara Institute of Science and Technology, Ikoma, Japan, ry.eguchi@is.naist.jp} \hspace{8pt} Fukuhito Ooshita \footnote{Fukui University of Technology, Fukui, Japan, f-oosita@fukui-ut.ac.jp}  \hspace{8pt}Michiko Inoue \footnote{Nara Institute of Science and Technology, Ikoma, Japan, kounoe@is.naist.jp} \hspace{8pt} S\'ebastien Tixeuil \footnote{Sorbonne Universit\'e, CNRS, LIP6, Institut Universitaire de France, Paris, France \\
Sebastien.Tixeuil@lip6.fr}}

\date{}

\begin{document}

\maketitle

\begin{abstract}
    In this paper, we revisit the problem of classical \textit{meeting times} of random walks in graphs. In the process that two tokens (called agents) perform random walks on an undirected graph, the meeting times are defined as the expected times until they meet when the two agents are initially located at different vertices. A key feature of the problem is that, in each discrete time-clock (called \textit{round}) of the process, the scheduler selects only one of the two agents, and the agent performs one move of the random walk. In the adversarial setting, the scheduler utilizes the strategy that intends to \textit{maximizing} the expected time to meet.
    In the seminal papers \cite{collisions,israeli1990token,tetali1993simult}, for the random walks of two agents, the notion of \textit{atomicity} is implicitly considered. That is, each move of agents should complete while the other agent waits. In this paper, we consider and formalize the meeting time of \textit{non-atomic} random walks. In the non-atomic random walks, we assume that in each round, only one agent can move but the move does not necessarily complete in the next round. In other words, we assume that an agent can move at a round while the other agent is still moving on an edge. For the non-atomic random walks with the adversarial schedulers, we give a polynomial upper bound on the meeting times. 
\end{abstract}

\clearpage

\section{Introduction}
In the process that two tokens (called agents) perform random walks on an undirected graph, the classical \textit{meeting times} are defined as the expected times until they meet when the two agents are initially located at different vertices~\cite{bshouty1999meeting, collisions,israeli1990token,tetali1993simult,tetali1991random}. A key feature of the meeting times is that, in each discrete time-clock (called \textit{round}) of the process, the scheduler selects only one of the two agents, and the agent performs one move of the random walk\footnote{In the literature, the term "meeting time" also be referred to represent the expected time to meet in the process that the two agents move randomly in each round~\cite{hassin1999distributed,cooper2013coalescing,oliveira2019random,kanade2023on}.}. Several schedulers are considered, namely, random scheduler where the scheduler selects an agent randomly; 'angel' scheduler with the intent of \textit{minimizing} the expected time before the tokens meet; or adversarial scheduler with the intent of \textit{maximizing} the expected time\cite{tetali1993simult}. The schedulers have \textit{strategies}. According to a strategy, a scheduler chooses an agent to move in the current configuration. 
Note that, the schedulers do not know the future moves of the agent, that is, it does not have any prior information about the random bits used by the agents. 
For the adversarial scheduler, Coppersmith et al.~\cite{collisions} show an upper-bound of the meeting times of each initial configuration, and after that, Tetali and Winkler~\cite{tetali1993simult} give the exact characterization of the meeting time of each initial configuration. (The specifications are presented in Subsection \ref{subsec:related}.) The process they consider, however, assumes the \textit{atomicity} of two independent random walks. That is, while an agent moves, the other agent has to wait until the move completes. Specifically, we define two random walks are \textit{atomic}, if in each (atomic) round, (1) only one agent moves, and (2) at the next round, the moving agent completes the move and reaches the next vertex.

In this paper, we consider the \textit{non-atomic} random walks of the two agents. In the non-atomic random walks, we assume that in each time step, (1) only one agent can move but (2) the move does not necessarily complete in the next round. In other words, we assume that an agent can be chosen to move at a round while the other agent $b$ is still moving on an edge to the goal vertex. To define such behavior, we consider \textit{subdivided graph} $\tilde{G}$ of the original graph $G = (V, E)$. The graph $\tilde{G}$ is produced by subdividing each edge $e \in E$ (that is, adding another vertex in the intermediate point of each edge). We say the vertices in $V$ are \textit{original vertices}, and the added vertices \textit{intermediate vertices}. Intuitively, the situations in which an agent at some intermediate vertex in $\tilde{G}$ are interpreted as the ones in which the agent is traversing the corresponding edge in $G$. Therefore, the meetings at some intermediate vertex in $\tilde{G}$ are interpreted as the meetings in an edge in $G$. 

Suppose that the two agents are initially located at original vertices $x, y \in V$ at round $r$. The non-atomic random walks under any scheduler proceed as follows: When the agent located at $x$ is chosen to move, it determines to move to adjacent vertex $w \in N_G(x)$. Then, it moves to the intermediate vertex between $x$ and $w$ \textit{with direction} from $x$ to $w$ at the current round $r$. So the next round $r+1$, the scheduler chooses an agent from the configuration that agents are located at $y$ and intermediate vertex between $x$ and $w$  with direction from $x$ to $w$. At the round $r+1$, if the adversary chooses the agent in the intermediate vertex, then the agent reaches vertex $w$ at round $r+2$. 
Otherwise, if at round $r+1$, the agent at $y$ is chosen to move, it also reaches an intermediate vertex with the corresponding direction. 
In the executions, the scheduler repeatedly chooses one of the agents. 
For the non-atomicity, we do not impose any restrictions for reaching the goals of original vertices from the moves of intermediate vertices. For example, in an extreme case, the adversary can only choose an agent once to locate the agent to the intermediate vertex, and after that, it can repeatedly move the other agent.

When an agent is located at an intermediate vertex, then the state of the agent is interpreted as traversing the corresponding edge in the original graph $G$. Also, similar to the atomic case, the agents do not return to the original vertex they left in the random walks. Therefore, when the agent located at an intermediate vertex is chosen to move, the direction for the move is kept. In other words, when the agent is located at an intermediate vertex between $v$ and $w$ with direction from $v$ to $w$, then it must reach the vertex $w$ (not $v$). 
Note also that, if the meetings can occur only when the agents are located at the same original vertex, then any meeting does not occur in the non-atomic random walks. For example, the scheduler at first moves an agent to the intermediate vertex and repeatedly moves another agent. By doing so, the agents never meet at the original vertices. Therefore, we need to allow the agents to meet in the intermediate vertices as well as the original vertices.Allowing agents to meet on intermediate vertices hints that previous results on random walks in the atomic case may be simply expanded to the non-atomic case. However, this is not the case:
\begin{itemize}
\item Considering a random walk on the subdivided graph (illustrated in part (a) of Fig. \ref{fig:formalization}) does not take into account that in our setting, since agents may not return back to the original node they left when at an intermediate node.
\item Considering a random walk on the directed graph induced by the intermediate node direction constraints (part (b) in Fig. \ref{fig:formalization}) does not take into account that for two original neighboring vertices $u$ and $v$, the intermediate vertices $\pi_{\{u, v\}}$ and $\pi_{\{v, u\}}$ must be a single vertex (otherwise, any meeting cannot occur at an intermediate node). 
\end{itemize}

\begin{figure}[ht]
 \centering
 \includegraphics[keepaspectratio, scale=0.5]
      {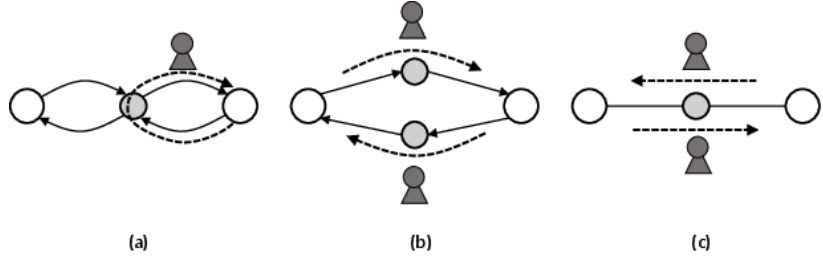}
 \caption{In this figure, the white circles represent the original nodes, and the gray circles represent the intermediate nodes. The objects with a black circle combined with a triangle represent the agents. \textbf{(a)} Random walks in the subdivided graphs. In this case, returning back to the original node they left is unavoidable. \textbf{(b)} Two intermediate nodes. In this case, the agents cannot meet when they enter the two intermediate nodes in different directions. \textbf{(c)} Our definition. The intermediate node is unique, and the agents have a direction as their state(Precise definition is presented in Section \ref{sec:preliminaries}).}
 \label{fig:formalization}
\end{figure}

\subsection{Our contribution}
The first contribution is the formalization of the meeting time of non-atomic random walks by two agents. For the formalization, we introduce the notion of non-atomic moves of agents in the subdivided graphs $\tilde{G}$. 
We then show that our definitions match the intuitions described above by formally proving the impossibility to meet when we restrict the meeting points to the original vertices in $V$.
By the impossibility, we relax the assumption such that the agents can meet at the intermediate vertices as well as the original vertices. We assume that the agents can meet when they are located at the same intermediate vertex regardless of their direction. That is, they can meet if they are located at the same intermediate vertex in the same direction or opposite directions.

Then, we prove an upper bound of the meeting times of the non-atomic random walks. Specifically, based on the proof arguments in~\cite{collisions}, we extend their proofs and show the following theorem. For the adversarial scheduler, let $\tilde{M}_{G} (x, y)$ denote the worst meeting time of non-atomic random walks from initial positions $x, y \in V$. 
Also, let $H(x, y)$ be the hitting time from $x$ to $y$, which is the expected time to reach $y$ from $x$ by an agent for $x, y \in V$ in $G$.

\begin{restatable}{theorem}{upperbound}
    \label{thrm:upperbound}
    For a pair of $x, y \in V$, suppose that the agents are initially located at vertices $x, y$ and conduct the non-atomic random walks under the strategy of the adversarial scheduler. Then, there is a pair of vertices $t, u \in V$ such that
    \begin{equation*}
        \tilde{M}_G(x, y) \le 2 \left( H(x, y) + H(y, u) - H(u, y) + d_u - 1 + \sum_{z \in N_G(u) \setminus {t}} H(z, u) \right)
    \end{equation*}
    holds, where $d_u$ is the degree of the original vertex $u \in V$.
\end{restatable}
The vertices $t, u$ in Theorem \ref{thrm:upperbound} are special in the sense that the state $s_{tu}$ is \textit{hidden}, whose precise definition appears in Subsection \ref{subsec:hidden}. Also, since the hitting times have a polynomial upper bound of $O(n^3)$, Theorem \ref{thrm:upperbound} also shows the polynomial upper bound of $O(n^4)$ on the worst meeting times of the non-atomic random walks.

\subsection{Related work}
\label{subsec:related}
The meeting times of (atomic) random walks were first presented by Israeli and Jalfon~\cite{israeli1990token}. In the paper, they introduced a token-management scheme for self-stabilizing mutual exclusion. In the initial configuration of the scheme, multiple tokens (i.e., agents) can exist, and processors (i.e., vertices) send the tokens to their adjacent processors randomly. If some tokens move to the same processor, then the tokens are merged into one token. 
For the scheme, it is required that tokens eventually are eventually merged into one token. In particular, the authors showed that, in a ring of $n$ vertices, the meeting time is $O(n^2)$, and also that, in general graphs, an exponential upper bound $O((\Delta - 1)^{D-1})$ exists, where $\Delta$ is the maximum degree of the graph, and $D$ is the diameter of the graph.

Tetali and Winkler~\cite{tetali1991random} proved the upper bound of the adversarial meeting times of random walks, specifically as follows. For any connected graph $G$, $M_G(x,y) \le H(x, y) + H(y, z) - H(z, y)$ holds for the initial positions $x, y$, where $z$ is special vertex called the \emph{hidden} vertex. In more detail, $z$ satisfies $H(z, v) \le H(v, z)$ for each $v \in V$. Hence the meeting times $M_G$ are upper bounded by at most twice the worst hitting time of the graph $G$ in the worst adversarial settings. 
 Then, Coppersmith, Tetali, and Winkler~\cite{collisions} showed another upper bound: the worst meeting time is upper bounded by $(\frac{4}{27} + o(1))n^3$ for any graphs and any initial positions. 

Tetali and Winkler~\cite{tetali1993simult}  provided the exact characterization of the meeting time of each initial configuration. 
The specification of the characterization is as follows: Let $c$ be the vertex such that, for any vertex $v \in V$, $H(v, c) \le H(c, v)$ holds. Let $q_G(w; x, y)$ denote the probability that agents initially located at $x, y$ first meet at vertex $w$ and let $Z(x)$ be another potential function that $Z(x) = H(c, x) - H(x, c)$. Then, they show 
\begin{equation*}
    M_G(x, y) = \Phi(x, y) - \sum_{w \in V} q_G(w; x, y) [Z(z) - Z(w)],
\end{equation*}
where $\Phi(x, y) = H(x, y) + H(y, z) - H(z, y)$ is introduced by Coppersmith, Tetali, and Winkle~\cite{collisions}.

The meeting times of (atomic) random walks by $k$ agents for $k \ge 2$ were examined by Bshouty et al.~\cite{bshouty1999meeting}. In the paper, they showed that the meeting times of multiple random walks have an upper bound in terms of the meeting times of fewer random walks. The meeting time of random walks by $k$ agents is the expected number of rounds to merge the agents into one agent. 
Quantities related to a random walk in graphs such as hitting time, cover time, and commute time are well studied~\cite{graham1990maximum, feige1995tight, kahn1989cover, tetali1991effective} and the interested reader can refer to the survey of Lovasz~\cite{lovasz1993random}. Interestingly, Brightwell and Winkler showed \cite{graham1990maximum} that the worst hitting time among all graphs is $O(n^3)$, which appears in the so-called lollipop graphs. The relation between random walks in graphs and electrical networks has also been investigated~\cite{tetali1991effective}.

The related problem concerning the meeting time of the non-atomic random walks is \textit{asynchronous rendezvous} of two agents~\cite{czyzowicz2012how, marco2005asynchronous,bampas201980, guilbault2011asynchronous}. 
The asynchronous rendezvous problem is similar problem to the non-atomic meeting time, in the sense that the adversary controls the speed of the agents (or the rounds to reach the destination of a move).
 The notion of subdivided graphs was introduced to study asynchronous rendezvous by Bampas et al.~\cite{bampas201980}.

\subsection{Organization of paper}
In Section \ref{sec:preliminaries}, we explain the definitions and notations for the meeting times of non-atomic random walks. In Section \ref{sec:impposibility}, we prove the impossibility that the agents cannot meet with the assumption that they are only assumed to meet in the original vertices. In Section \ref{sec:upperbound}, we show an upper bound for meeting times of non-atomic random walks. In Section \ref{sec:discussion}, we examine the upper bound in several graph classes such as lines and rings, and complete graphs. Finally, we conclude this paper in Section \ref{sec:conclusion}.

\section{Preliminaries}
\label{sec:preliminaries}

Let $G = (V(G), E(G))$ be a connected undirected graph with $n$ vertices and $m$ edges. Let $\Gtilde$ denote the graph produced by subdividing each edge of $G$ into two parts. Precisely, we define $\Gtilde = (\Vtilde , \Etilde)$. The vertex set is defined by $\Vtilde = V(G) \cup \Vint$, where $\Vint = \{ \pi_{\{ v, w \} } \mid (v,w) \in E(G) \}$. The edge set is defined by $\Etilde = \{ (v, \pi_{\{v, w\}}) \mid v \in V(G) \wedge \pi_{\{v, w\}} \in \Vint \}$. Note that $\pi_{\{v, w\}} = \pi_{\{w, v\}}$ holds. Also, we say vertices in $V(G)$ are \textit{original vertices}, and vertices in $\Vint$ are \textit{intermediate vertices}. 

Each state of an agent is an element from a set $V \cup \Sint$, where $\Sint = \{ s_{xy} \mid \pi_{\{x, y\}} \in \Vint \}$. The state $s_{xy}$ represents that the agent is located at $\pi_{\{x, y\}}$ and it has the direction from $x$ to $y$. That is, in the original graph $G$, the agent with state $s_{xy}$ is heading to the vertex $y$ from $x$ in the edge $(x, y)$ in $G$. Note that for each $\pi_{\{v, w\}}$, there are two states $s_{xy}$ and $s_{yx}$, and $s_{xy} \neq s_{yx}$ holds since the directions are different. The states from $V$ represent that the current position of the agent is an original vertex of $G$. We call the states from $V$ \textit{original states}. We define states from $\Sint$ are \textit{intermediate states}. For each state $s$, we define the following bar operation: if $s \in V$, then $\overline{s} = s$, and if $s = s_{xy} \in \Sint$, then $\overline{s_{xy}} = s_{yx}$.
Next, we define the adjacent states of each state. 
The set of adjacent states $N_{as}(s)$ for a state $s \in V \cup \Sint$ is defined that (1) $N_{as}(s) = \{ s_{vw} \in \Sint \mid w \in N_G(v) \}$ if $s = v$ (original), and (2) $N_{as}(s) = \{w\}$ if $s = s_{vw} \in \Sint$ (intermediate).
We also define $d_s = |N_{as}(s)|$.
The set of configurations $\mathcal{C}$ is defined by $(V \cup \Sint)^2$. The elements of a configuration correspond to the states of the two agents. 

The computation proceeds in discrete time $r = 0,1,2, \dots$, which is called \textit{rounds}. In each round $r$, the adversary determines which agent moves for a configuration $c$ in the round. The selected agent at state $s$ moves to an adjacent state $s' \in N_{as}(s)$ in $\Gtilde$ with probability $1/d_s$. Note that if the state is intermediate state $s_{vw}$, then the next state is uniquely determined to $w$. Note that even if we consider the non-atomic moves of the agents, these moves of the agents are assumed to be complete in the current round, and in the next round the moving agent should be at the next state in $V \cup \Sint$. 

The adversary chooses an agent in the current configuration according to a strategy $S_{\tilde{G}}$ for a subdivided graph $\tilde{G}$. 
The strategy $S_{\tilde{G}}$ is a function $S_{\tilde{G}}: \mathcal{C} \rightarrow [0,1]$, where $\mathcal{C}$ is a set of configurations.
The expression $S_{\tilde{G}}(c) = p$ for $c = (s, s') \in (V \cup \Sint)^2$ and $p \in [0,1]$ represents that the adversary moves an agent in the state $s$ with probability $p$ (otherwise it moves the other agent in the state $s'$) in the current configuration $c$. 

We say that the agents meet, if the execution reaches a configuration $c = (s, s')$ such that $s = s'$ or $s = \overline{s'}$. In other words, the agents meet if they are located at the same original vertices or intermediate vertices regardless of the direction of the agents. We define $\tilde{M}_S (x, y)$ as the expected rounds for non-atomic random walks starting at initial position $x, y$ to meet when the adversary adopts the strategy $S$. $\Mtld_G(x, y)$ denotes the worst meeting time of non-atomic random walks starting at $x, y \in V$ for all strategies in $\tilde{G}$. 

\section{Impossibility Results}
\label{sec:impposibility}
In this section, using our definition we show the impossibility that the agents cannot meet when we restrict the meeting points to original vertices. 
\begin{theorem}
    If the meeting cannot occur on intermediate vertices, there exists an adversarial strategy such that agents never meet at an original vertex.
\end{theorem}

\begin{proof}
    We specify the strategy as follows. For each $v \in V$ and each $s_{uv} \in \Sint$ for $u \in N_G(v)$, the strategy moves the agent at $v$ with probability $1$ in the configuration $(v, s_{uv})$, and moves agents arbitrary in other configurations.
    Obviously, to meet at an original vertex, the agents should reach the configuration $(v, s_{uv})$ for some $v$ and $s_{uv}$. However, in the next configuration, the agents cannot meet by the strategy, since by the assumption of meeting, the agents at $s_{vu}$ and $s_{uv}$ cannot meet. 

\end{proof}

Observe that weakening the power of the scheduler does not help. Even if the scheduler is $2$-fair (in such a strategy, each agent is selected infinitely often, and between any two selections of an agent, any other agent is selected at most twice), it remains impossible to obtain meeting if they can only occur on original vertices. For example, consider an alternating strategy (the scheduler alternates between two consecutive activations of the two agents, except for the first activations in the strategy, where the first time an agent is activated, it is activated only once). Then, after the first activation, the first agent is on an intermediate node. After the second agent is activated, both agents are on intermediate nodes. Then, the first agent is activated twice, and both agents remain on intermediate nodes. The selection continues so that both agents are on intermediate nodes at the end of each activation.

\section{An upper bound for non-atomic meeting time}
\label{sec:upperbound}
In this section, we show the upper bound for non-atomic meeting time. At first, we define the hitting times between states in $\Gtilde$, which is a generalization of the hitting time between vertices in $G$. In the following, we say the generalized version of the hitting times \textit{extended hitting times}. We also show a property of the extended hitting time, which we call triangle property in the following. It is also the generalized version of the triangle property shown in the arguments in \cite{collisions}. In our argument, the property is generalized in a bit tricky way to hold the latter arguments. Using the generalized triangle property, we can prove the existence of special states called \textit{hidden states} (Subsection \ref{subsec:hidden}). The hidden states allow us to introduce a potential function $\tilde{\Phi}$ for a pair of states (Subsection \ref{subsec:main}), and finally we show an upper bound on the non-atomic meeting times using the potential function $\tilde{\Phi}$.

\subsection{Hitting time of states and triangle property}
For a pair of states $s, s'$, we define that the extended hitting time $\tilde{H}(s, s')$ is the expected moves to reach the state $s'$ from the state $s$ by an agent in $\tilde{G}$. The moves of the agent are the same as the ones in the non-atomic moves in Section \ref{sec:preliminaries}, specifically as follows: if its state $s$ is in $V$, then it moves to $s_{sw} \in N_{as}(s)$ with probability $1/d_s$ for each $w \in N_G(v)$; If its state $s$ is $s_{xy} \in \Sint$, then it moves to $y$ as a move. 

If $s, s' \in V$ holds, then $\tilde{H}(s, s')$ is twice of the value of the original hitting time $H(s, s')$ in $G$, that is, $\tilde{H}(s, s') = 2H(s, s')$. This is because, in the moves in $\tilde{G}$, the agent should move twice to traverse an edge corresponding to $G$. If the starting point is intermediate, that is, $s = s_{xy} \in \Sint$ and $s' \in V$, then we have $\tilde{H}(s_{xy}, s') = 1 + \tilde{H}(y, s') = 1 + 2H(y, s')$. The extended hitting time in this case can be also calculated by the original hitting time. The remaining case is that the goal state is an intermediate state, that is, $s' = s_{xy}$ for $s_{xy} \in \Sint$. In this case, to reach $s_{xy}$ the agent should visit the original vertex $x$. Therefore, the following equality holds by the linearity of expectation: $\tilde{H}(s, s_{xy}) = \tilde{H}(s, x) + \tilde{H}(x, s_{xy})$. Therefore, we should calculate the value of $\tilde{H}(x, s_{xy})$ for any $x$ and $s_{xy} \in N_{as}(x)$. 

\begin{lemma}
    \label{lma:neighborhit}
    For each $x \in V$ and $s_{xy} \in \Sint$, we have
    \begin{align*}
        \tilde{H}(x, s_{xy}) = 2 d_x - 1 + \sum_{z \in N_G(x) \setminus \{y\}} \tilde{H}(z, x).
    \end{align*}
\end{lemma}

\begin{proof}
    Let $\tilde{H}(x, s_{xy}) = T$. At vertex $x$, the agent reaches the state $s_{xy}$ with the probability $1/d_x$. Otherwise, it moves $z$ for $z \in N_G(x) \setminus \{y\}$ with the same probability with two moves. After the latter case, the agent should return to the vertex $x$ for reaching the state $s_{xy}$. Therefore, the value can be written the following recursive formula:
    \begin{equation*}
        T = \frac{1}{d_x} \left( 1 + \sum_{z \in N_G(x) \setminus \{ y \}} (2 + \tilde{H}(z, x) + T) \right). 
    \end{equation*}
    Therefore we have
    \begin{eqnarray*}
        d_x \cdot T &=& 1 + 2d_x -2 + \left( \sum_{z \in N_G(x) \setminus \{ y \}} \tilde{H}(z, x) \right) + (d_x -1) \cdot T \\
        T &=&  2d_x - 1 + \sum_{z \in N_G(x) \setminus \{ y \}} \tilde{H}(z, x).
    \end{eqnarray*}

\end{proof}

Next, we prove the key property of the extended hitting times. Here we introduce the original triangle property of the original hitting times.

\begin{lemma}[From \cite{collisions}]
    \label{lma:originaltriangle}
    For any $x, y, z \in V$, we have \\
    $
        H(x,y) + H(y,z) + H(z, x) = H(x, z) + H(z, y) + H(y, x)
    $
\end{lemma}

Note that the left side of the equation is the expected time that a random walk starting at $x$ visits $y$ then visits $z$, and returns to $x$, similarly to the right side of the equation.
While there may exist multiple generalizations of the property,
we use the following generalization of the triangle property to establish the latter argument of our proof.

\begin{lemma}
    \label{lma:extendedtriangle}
    For states $x, y, z \in V \cup \Sint$, we have \\
    $
        \tilde{H}(x, \overline{y}) + \tilde{H}(y, \overline{z}) + \tilde{H}(z, \overline{x}) = \tilde{H}(x, \overline{z}) + \tilde{H}(z, \overline{y}) + \tilde{H}(y, \overline{x})
    $
\end{lemma}

\begin{proof}
    To prove the lemma, we use the function $f: V \cup \Sint \rightarrow V$, where $f(v) = v$ for $v \in V$ and $f(s_{ab}) = a$ for $s_{ab} \in \Sint$. Using the function, we can rewrite the extended hitting time as follows: For states $s, s' \in V \cup \Sint$, $\tilde{H}(s, s') = \tilde{H}(s, f(s')) + \tilde{H}(f(s'), s')$. 
    Using the function, we rewrite the left side of the equation in the lemma as
    \begin{eqnarray*}
        \tilde{H}(x, \overline{y}) + &\tilde{H}&(y, \overline{z}) + \tilde{H}(z, \overline{x}) \\
         &=& \tilde{H}(x, f(\overline{y})) + \tilde{H}(f(\overline{y}), \overline{y}) + \tilde{H}(y, f(\overline{z})) + \tilde{H}(f(\overline{z}), \overline{z}) \\
         && \, + \tilde{H}(z, f(\overline{x})) + \tilde{H}(f(\overline{x}), \overline{x}) \\
         &=& \tilde{H}(x, f(\overline{y})) + \tilde{H}(y, f(\overline{z})) + \tilde{H}(z, f(\overline{x})) \\
         && \, + \tilde{H}(f(\overline{x}), \overline{x}) + \tilde{H}(f(\overline{y}), \overline{y}) + \tilde{H}(f(\overline{z}), \overline{z}).
    \end{eqnarray*}

    Similarly, for the right side of the equation, we have
    \begin{eqnarray*}
        \tilde{H}(x, \overline{z}) + \tilde{H}(z, \overline{y}) + \tilde{H}(y, \overline{x}) &=& \tilde{H}(x, f(\overline{z})) + \tilde{H}(z, f(\overline{y})) + \tilde{H}(y, f(\overline{x})) \\
        && \, + \tilde{H}(f(\overline{x}), \overline{x}) + \tilde{H}(f(\overline{y}), \overline{y}) + \tilde{H}(f(\overline{z}), \overline{z}).
    \end{eqnarray*}

    Therefore it is sufficient to show that $\tilde{H}(x, f(\overline{y})) + \tilde{H}(y, f(\overline{z})) + \tilde{H}(z, f(\overline{x})) = \tilde{H}(x, f(\overline{z})) + \tilde{H}(z, f(\overline{y})) + \tilde{H}(y, f(\overline{x}))$. We then define another function $g: V \cup \Sint \rightarrow V$, which returns $v$ if the input is $v \in V$, and $b$ if it is $s_{ab} \in \Sint$. Similarly for $f$, we have $\tilde{H}(s, s') = \tilde{H}(s, g(s)) + \tilde{H}(g(s), s')$. Thus this proof is reduced to if the following equation holds: $\tilde{H}(g(x), f(\overline{y})) + \tilde{H}(g(y), f(\overline{z})) + \tilde{H}(g(z), f(\overline{x})) = \tilde{H}(g(x), f(\overline{z})) + \tilde{H}(g(z), f(\overline{y})) + \tilde{H}(g(y), f(\overline{x}))$. For the functions, it holds that $g(x) = f(\overline{x})$  and $g(x)$ is an original state. Therefore, the last equation holds by the original triangle property of Lemma \ref{lma:originaltriangle}.
    
\end{proof}

\subsection{Hidden states}
\label{subsec:hidden}

We define the following relation $\le_{EHT}$ as follows: For any pair of states $s, s'$, $s \le_{EHT} s'$ holds if $\tilde{H}(s, \overline{s'}) \le \tilde{H}(s', \overline{s})$ holds. 
It is also the extension of the relation $\le_{HT}$ given in \cite{collisions}. In the proofs of \cite{collisions} the relation for the original hitting times is defined as follows: For any vertex $v, w \in V$, the relation $v \le_{HT} w$ holds if $\tilde{H}(v, w) \le \tilde{H}(w, v)$ holds. 
The main purpose of the relations is to prove the existence of the \textit{hidden} vertex (or state). The hidden vertex in the original argument is the minimum vertex in the relation. 
It is proven by showing the relation is transitive. We also prove the extended relation $\le_{EHT}$ is transitive, and prove the existence of the hidden state.

\begin{restatable}{lemma}{transitive}
    \label{lma:transitive}
    The relation $\le_{EHT}$ is transitive. As a consequence, there is a state $s \in V \cup \Sint$ such that for any $s' \in V \cup \Sint$, it holds that $s \le_{EHT} s'$.
\end{restatable}

\begin{proof}
    It is sufficient to show that for any states $x, y, z \in V \cup \Sint$, if $x \le_{EHT} y$ and $y \le_{EHT} z$ holds, then we have $x \le_{EHT} z$. By the assumption, we have $\tilde{H}(x, \overline{y}) \le \tilde{H}(y, \overline{x})$ and $\tilde{H}(y, \overline{z}) \le \tilde{H}(z, \overline{y})$. Applying the triangle property for the states $x, y, z$ of Lemma \ref{lma:extendedtriangle}, we have
    
    \begin{eqnarray*}
        \tilde{H}(x, \overline{y}) + \tilde{H}(y, \overline{z}) + \tilde{H}(z, \overline{x}) &=& \tilde{H}(x, \overline{z}) + \tilde{H}(z, \overline{y}) + \tilde{H}(y, \overline{x}) \\
        \tilde{H}(x, \overline{y}) - \tilde{H}(y, \overline{x})  +  \tilde{H}(y, \overline{z}) - \tilde{H}(z, \overline{y}) &=& \tilde{H}(x, \overline{z}) - \tilde{H}(z, \overline{x})\\
        \tilde{H}(x, \overline{z}) - \tilde{H}(z, \overline{x}) &\le& 0
    \end{eqnarray*}
    Therefore it holds that $\tilde{H}(x, \overline{z}) \le \tilde{H}(z, \overline{x})$, proving the lemma. 
\end{proof}

Then, we show that the hidden state(s) is intermediate. For any original state $v \in V$, we show that $s_{wv} \le_{EHT} v$ for any intermediate $s_{wv} \in \Sint$ for $w \in N_G(v)$. For any vertex $v$ and any $s_{wv}$, we have $\tilde{H}(s_{wv}, \overline{v}) = \tilde{H}(s_{wv}, v) = 1$ by definition, and $\tilde{H}(v, \overline{s}_{wv}) = \tilde{H}(v, s_{vw}) = 2d_v - 1 + \sum_{z \in N_{as}(v) \setminus \{w\}} \tilde{H}(z, v)$. Since the graph $\tilde{G}$ is connected, we have $d_v \ge 1$. Hence it holds that $\tilde{H}(v, s_{vw}) \ge 1$. Therefore for $v, s_{wv}$, we have $s_{wv} \le_{EHT} v$, as desired.

\begin{proposition}
    An intermediate state is hidden.
\end{proposition}

\subsection{Main argument}
\label{subsec:main}

Now we define a potential function $\tilde{\Phi}$. For each pair of states $x, y$ and a hidden state $s_{tu}$, we set the function $\tilde{\Phi}(x, y) = \tilde{H}(x, \overline{y}) + \tilde{H}(y, \overline{s_{tu}}) - \tilde{H}(s_{tu}, \overline{y})$. The function is derived using the triangle property for states $s, s', s_{tu}$ as follows:
\begin{eqnarray*}
    \tilde{H}(x, \overline{y}) + \tilde{H}(y, \overline{s_{tu}}) + \tilde{H}(s_{tu}, \overline{x}) &=& \tilde{H}(x, \overline{s_{tu}}) + \tilde{H}(s_{tu}, \overline{y}) + \tilde{H}(y, \overline{x}) \\
    \tilde{H}(x, \overline{y}) + \tilde{H}(y, \overline{s_{tu}}) - \tilde{H}(s_{tu}, \overline{y}) &=& \tilde{H}(y, \overline{x}) + \tilde{H}(x, \overline{s_{tu}}) - \tilde{H}(s_{tu}, \overline{x})
\end{eqnarray*}

Therefore it holds that $\tilde{\Phi}(x, y) = \tilde{\Phi}(y, x)$. Also, since the states $s_{tu}$ is hidden, we have that $\tilde{H}(z, \overline{s_{tu}}) - \tilde{H}(s_{tu}, \overline{z}) \ge 0$ for any state $z \in V \cup \Sint$. Thus we have the following proposition, which is implicitly used in Theorem \ref{thrm:potential}. 

\begin{proposition}
    \label{fact:phigezero}
    $\tilde{\Phi}(x, y) \ge 0$ for any $x, y \in V \cup \Sint$
\end{proposition}

Now we add some modifications for initial positions, define some notations, and present the preliminary propositions/lemmas for the main proof. At first, we extend the initial positions of the non-atomic meeting times, to the ones that include states from $V \cup \Sint$. That is, we define that the starting states of the agents include the intermediate states.  
In the following, we assume that any strategy $S$ contains the intermediate states as initial positions.
Next, we define the optimal (longest) and deterministic strategy. A strategy is deterministic, if for any pair $s, s' \in V \cup \Sint$, the moving agent is chosen with probability $1$. Also, an strategy $S$ is optimal, if for any pair $s, s' \in V \cup \Sint$, it holds that $\Mtld_G(s, s') = \Mtld_S(s, s')$. 

We also define a value of configurations called \textit{destination value}. For a configuration $(s_1, s_2)$, the value is defined 
according to the function $g$ defined in the proof of Lemma \ref{lma:extendedtriangle}. Recall that $g(s) = s$ if $s = v \in V$ and $g(s_{ab}) = b$ for $s_{ab} \in \Sint$. The destination value of the 
configuration $(s_1, s_2)$ is defined as $d(s_1, s_2) = dist(s_1, g(s_1)) 
+ dist(g(s_1), g(s_2)) + dist(s_2, g(s_2))$, where (1) $dist(s, g(s)) = 1$ if the state $s$ is intermediate, and (2) $dist(s, g(s)) = 0$ if $s \in V$, and (3) for $s, s' \in V$, $dist(s, s')$ is the hop-distance of the vertices in $\tilde{G}$. Observe that $d(s_1, s_2) \ge d(g(s_1), g(s_2))$ holds.

To prove the existence of the optimal and deterministic strategy, we first show the following lemma. Let $S(x,y)$ be the strategy that maximizes the non-atomic meeting time starting at $x,y \in V \cup \Sint$, i.e., $\Mtld_{S(x,y)}(x,y) = \Mtld_G(x,y)$ holds. Its proof is deferred to Appendix \ref{apd:meeteq}.

\begin{restatable}{lemma}{meeteq}
    \label{lma:meeteq}
    For any $x, y \in V$, suppose that $S(x, y)$ moves the agent with the state $x$ with positive probability $p > 0$. Then, we have that 
    \begin{equation*}
        \label{eqn:eq}
    \Mtld_{S(x,y)}(x,y) = 1 + \frac{1}{d_x} \sum_{z \in N_{as}(x)}\Mtld_G(z, y).
    \end{equation*}
\end{restatable}

\begin{lemma}
    \label{lma:optimal}
    For any $G$, there is an optimal and deterministic strategy $S^{\ast}$.
\end{lemma}

\begin{proof}
    
    Using $S(x,y)$, we define the strategy $S^{\ast}$ as follows: If the strategy $S(x,y)$ moves the agent at $x$ at initial configuration $(x,y)$ with the probability strictly greater than $0$, then the strategy $S^{\ast}$ moves $x$ at $(x,y)$. Otherwise if $S(x,y)$ moves $y$ with probability $1$, then $S^{\ast}$ moves $y$. The strategy $S^{\ast}$ is created by conducting the above operation for all pairs of states $(x, y)$. Obviously, the strategy is deterministic. We argue that such strategy $S^{\ast}$ is optimal.

    Towards the contradiction, suppose that $S^{\ast}$ is not optimal. Then there is at least one pair of states $x, y$ such that $\Mtld_G(x,y) - \Mtld_{S^{\ast}}(x,y) > 0$ holds. Let $\alpha$ be the maximum value of $\Mtld_G(x,y) - \Mtld_{S^{\ast}}(x,y)$ among such pairs. We choose a pair $x, y$ such that they attain $\alpha$ and have the minimum destination value $d(x,y)$ among the pairs that attain the value $\alpha$. It holds that $x \neq y$, since if $x = y$ then $\Mtld_G(x, x) = \Mtld_{S^{\ast}}(x, x) = 0$. Assume that in the strategy $S(x,y)$ the agent in the state $x$ is moved with positive probability $p > 0$. Therefore, $S^{\ast}$ moves the agent in the state $x$ with probability one. Hence, by averaging the moves of the agent among the neighbors of the state $x$, we have
    \begin{equation}
        \label{eqn:aves}
        \Mtld_{S^{\ast}}(x,y) = 1 + \frac{1}{d_x} \left( \sum_{z \in N_{as}(x)}\Mtld_{S^{\ast}}(z, y) \right).
    \end{equation} 
    
    Using Lemma \ref{lma:meeteq} and Equation (\ref{eqn:aves}), we can derive the following contradiction,
    \begin{eqnarray*}
        \Mtld_{S(x,y)}(x,y) &=& 1 + \frac{1}{d_x}\sum_{z \in N_{as}(x)}\Mtld_G(z, y) \hspace{100pt}\text{(by Lemma \ref{lma:meeteq})}\\
        &=& 1 + \frac{1}{d_x}\sum_{z \in N_{as}(x)} \left( \Mtld_{S^{\ast}}(z, y) + \alpha(z, y) \right) \\
        &<& 1 + \frac{1}{d_x}\sum_{z \in N_{as}(x)} \left( \Mtld_{S^{\ast}}(z, y) \right) + \alpha \\
        &=& \Mtld_{S^{\ast}}(x, y) + \alpha \hspace{145pt}\text{(by Equation \ref{eqn:aves})}\\
        &=& \Mtld_G(x,y) = \Mtld_{S(x,y)}(x,y),
    \end{eqnarray*}
    where $\alpha(z,y) = \Mtld_{G}(z,y) - \Mtld_{S(z,y)}(z,y)$. The strict inequality holds because of the $(x,y)$ choice. That is, there is at least one adjacent state $z' \in N_{as}(x)$ such that $d(z, y) < d(x, y)$ by the definition of $d$. Therefore, at such a configuration $(z',y)$ we have $\alpha(z', y) < \alpha$.

    The other case is that in the strategy $S(x, y)$, the adversary moves the agent in the state $y$ with probability $1$. In this case, we can directly have $\Mtld_{S(x,y)}(x,y) = 1 + \frac{1}{d_y} \sum_{z \in N_{as}(y)}\Mtld_G(x, z)$. Also, $\Mtld_{S^{\ast}}(x,y) = 1 + \frac{1}{d_y} \left( \sum_{z \in N_{as}(y)}\Mtld_{S^{\ast}}(x, z) \right)$ holds. Therefore, we can derive a contradiction similarly, proving the lemma. 
\end{proof}

For the extended hitting times, the following proposition holds.

\begin{proposition}
    \label{proposition:averagehit}
    For any pair of states $s_1, s_2 \in V \cup \Sint$ such that $s_1 \neq s_2$, we have 
    \begin{equation*}
        \tilde{H}(s_1, s_2) = 1 + \frac{1}{d_{s_1}} \sum_{z \in N_{as}(s_1)} \tilde{H}(z, s_2).
    \end{equation*}
\end{proposition}
Thus we have the following equations for the potential function $\tilde{\Phi}$, whose proof is deferred to Appendix \ref{apd:averagephi}.

\begin{restatable}{lemma}{averagephi}
    \label{lma:averagephi}
    For any pair of states $s_1, s_2$ such that $s_1 \neq s_2$, we have
    \begin{equation*}
        \tilde{\Phi}(s_1, s_2) = 1 + \frac{1}{d_{s_1}} \sum_{z \in N_{as}(s_1)} \tilde{\Phi}(z, s_2)
        = 1 + \frac{1}{d_{s_2}} \sum_{z \in N_{as}(s_2)} \tilde{\Phi}(z, s_1).
    \end{equation*}
\end{restatable}

A similar proposition holds for the non-atomic meeting times.

\begin{proposition}
    \label{prop:averagemeet}
    For any pair of states $s_1, s_2$ such that $s_1 \neq s_2$ and an optimal and deterministic strategy $S^{\ast}$, suppose that the agent with state $s_1$ is moved in the configuration $(s_1, s_2)$ by the strategy. Then, we have
    \begin{equation*}
        \Mtld_G(s_1, s_2) = 1 + \frac{1}{d_{s_1}} \sum_{z \in N_{as}(s_1)} \Mtld_G(z, s_2).
    \end{equation*}
    Otherwise, if the strategy moves the agent with state $s_2$, then we have
    \begin{equation*}
        \Mtld_G(s_1, s_2) = 1 + \frac{1}{d_{s_2}} \sum_{z \in N_{as}(s_2)} \Mtld_G(z, s_1).
    \end{equation*}
\end{proposition}
Note that these two cases exclusively hold, that is, both equations do not hold at a time. This fact holds since in the case that the adversary moves the agent $x$, we cannot take the average among the adjacent states of $y$.

Finally, we can claim the main argument.

\begin{theorem}
    \label{thrm:potential}
    Let $G$ be any connected and undirected graph, and let $s_{tu}$ be a hidden state of non-atomic random walks of $\tilde{G}$. Then, for every pair of states $x, y$, we have
    \begin{equation*}
        \Mtld_G(x,y) \le \tilde{\Phi}(x, y),
    \end{equation*}
    where
    \begin{equation*}
        \tilde{\Phi}(x,y) = \tilde{H}(x, \overline{y}) + \tilde{H}(y, \overline{s_{tu}}) - \tilde{H}(s_{tu}, \overline{y}).
    \end{equation*}
\end{theorem}

\begin{proof}
    To prove the theorem by contradiction, assume that there is a pair of states $x, y$ such that $\Mtld_G(x, y) - \tilde{\Phi}(x, y) > 0$. Let $\beta_{max}$ be the maximum value of such differences. Let $\beta(x, y)$ be the difference $\Mtld_G(x, y) - \tilde{\Phi}(x, y)$ at the configuration $(x, y)$. We choose a configuration $(x, y)$ such that it obtains minimum $d(x, y)$ among the configurations that achieve $\beta_{max}$. Without loss of generality, suppose that the strategy $S^{\ast}$ moves the agent in the state $x$. Since $\Mtld_G(x, x) = 0$ and $\tilde{\Phi}(x, x) \ge 0$ by Proposition \ref{fact:phigezero}, it holds that $x \neq y$. Using the average argument, we have the following contradiction:
    \begin{eqnarray*}
        \Mtld_G(x,y) &=& \tilde{\Phi}(x, y) + \beta_{max} \\
            &=& 1 + \frac{1}{d_{x}} \sum_{z \in N_{as}(x)} \tilde{\Phi}(z, y) + \beta_{max} \\
            &>& 1 + \frac{1}{d_{x}} \sum_{z \in N_{as}(x)} \left( \tilde{\Phi}(z, y) + \beta(z, y) \right) \\
            &=& 1 + \frac{1}{d_{x}} \sum_{z \in N_{as}(x)} \Mtld_G(z, y)
            = \Mtld_G (x, y).
    \end{eqnarray*}
    The second equality uses Lemma \ref{lma:averagephi}, and the last equality uses Proposition \ref{prop:averagemeet}. The inequality is strict, since there is an adjacent state $z' \in N_{as}(x)$ such that $d(z',y) < d(x,y)$, therefore we have $\beta(z', y) < \beta_{max}$.

\end{proof}

As a corollary of Theorem \ref{thrm:potential}, we can prove Theorem \ref{thrm:upperbound}, however, whose proof is deferred to Appendix \ref{apd:upperbound} due to the space constraint.
\section{Discussion}
\label{sec:discussion}
In this section, we examine the upper bounds for several graph classes, namely graphs with bounded degrees including lines and rings, and complete graphs. We also consider the general graphs. Let $\Delta$ be the maximum degree of the vertices in $V$, and $H_G = \max_{x, y \in V} H(x, y)$.

\begin{itemize}
    \item In general graphs, we have the following upper bound. Since in the last summation of the upper bound of Theorem \ref{thrm:upperbound}, the number of sums is upper bounded by the max degree of the initial positions $x, y$, we have the general upper bound of $O(\Delta H_G)$. Also, $H_G = O(n^3)$ holds for any graph $G$, which is shown by the paper \cite{graham1990maximum}, we have $O(\Delta n^3) = O(n^4)$. For the meeting time of atomic random walks, in the paper \cite{tetali1991random} the authors show that $M_G(x, y) = O(n^3)$ for any graph $G$ and any initial positions $x, y$.
    
    \item For any graph $G$ with bounded degrees for any initial position $x, y$, the meeting time $\Mtld (x, y)$ is bounded by $O(H_G)$. Especially in the lines and rings, since the hitting time $H(z, u)$ for $z \in N_G(u)$ is $O(n)$, we have $\tilde{M}_G(x, y) \le 4H_G + O(n)$. 
    \item The complete graphs are an instance in which the meeting times are different between atomic and non-atomic random walks. In the original (atomic) random walks, the meeting times for the complete graph with $n$ vertices is $\Theta(n)$, while the one for the non-atomic random walk is $\Theta(n^2)$. For the original random walks, the $O(n)$-upper bound is derived by the original upper bound in \cite{collisions}. Also, $\Omega(n)$-lower bound is given by the following strategy: for the agents starting at different vertices, the strategy repeatedly chooses the same agent until they meet. Obviously, the expected time to meet is $\Theta(n)$ with the adversary using the strategy. Since there is a strategy that obtains $\Theta(n)$ time to meet, the meeting time of the worst strategy is at least $\Omega(n)$.
Similarly, in the non-atomic random walks, an $O(n^2)$ upper bound of the meeting time is given by calculating the potential function $\tilde{\Phi}$. Since any intermediate state is hidden by the symmetry of the topology, we can choose any intermediate state as a hidden state, suppose $s_h$. For any pair of intermediate states $s_{ab}, s_{cd}$, we have $\tilde{H}(s_{ab}, s_{cd}) = 1 + \tilde{H}(b,c) + \tilde{H}(c, s_{cd}) = O(n + \tilde{H}(c, s_{cd})) = O(n^2)$. Also, the upper bound is derived by the following strategy: we let the agents start at $s_{ab}$ and $s_h$, and the strategy repeatedly moves the agent starting at $s_{ab}$ until they meet. This takes expected $\Theta(n^2)$ time to meet, and as a consequence, the meeting time of the worst strategy is at least $\Omega(n^2)$.
\end{itemize}

\section{Conclusion}
\label{sec:conclusion}
In this paper, we revisit the adversarial meeting time of the random walks by two agents and consider the non-atomic version of the random walks. For the extended version of the random walks, we give a new upper bound of the worst-case expected time to meet in a given graph $G$. 
We also show for atomic and non-atomic random walks, the meeting times are different in the complete graphs. 

\bibliographystyle{plain}
\bibliography{arxiv}

\appendix

\section{Omitted Proofs}

\subsection{Proof of Lemma \ref{lma:meeteq}}
\label{apd:meeteq}
\meeteq*

\begin{proof}
    In the strategy $S(x, y)$, the adversary moves the agent in the state $x$ with probability $p$ or the one in the state $y$ with probability $1-p$ for $p > 0$. Since $\Mtld_{S(x, y)}(x, y) = \Mtld_G(x, y)$ holds, we have
    \begin{equation}
        \label{eqn:avesxy}
        \Mtld_{S(x,y)}(x,y) = 1 + \frac{p}{d_x} \left( \sum_{z \in N_{as}(x)}\Mtld_G(z, y) \right) + \frac{1 - p}{d_y} \left( \sum_{z \in N_{as}(y)}\Mtld_G(z, x) \right).
    \end{equation} 
    For the proof, we first show that $\Mtld_{S(x,y)}(x,y) \ge 1 + \frac{1}{d_x} \sum_{z \in N_{as}(x)}\Mtld_G(z, y)$. 
    Since $S^{\ast}$ moves the agent with the state $x$ with probability $1$, we have  $\Mtld_{S^{\ast}}(x,y) = 1 + (1/d_x)\sum_{z \in N_{as}(x)}\Mtld_G(z, y)$. Also, since $S(x, y)$ is the worst strategy when the initial position is $x, y$, we have $\Mtld_{S^{\ast}}(x, y) \le \Mtld_{S(x, y)}(x, y)$. Therefore, $\Mtld_{S(x,y)}(x,y) \ge 1 + \frac{1}{d_x} \sum_{z \in N_{as}(x)}\Mtld_G(z, y)$ holds.

    Then, we show that $\Mtld_{S(x,y)}(x,y) \le 1 + \frac{1}{d_x} \sum_{z \in N_{as}(x)}\Mtld_G(z, y)$. To prove it by contradiction, suppose that
    $
        \Mtld_{S(x,y)}(x,y) > 1 + \frac{1}{d_x} \sum_{z \in N_{as}(x)}\Mtld_G(z, y).
    $
    Then by applying the equation (\ref{eqn:avesxy}) to the assumption, we have
    \begin{eqnarray*}
        \Mtld_{S(x,y)}(x,y) &>& 1 + \frac{1}{d_x}\sum_{z \in N_{as}(x)}\Mtld_G(z, y) \\
        \frac{p}{d_x} \left( \sum_{z \in N_{as}(x)}\Mtld_G(z, y) \right) &+& \frac{1 - p}{d_y} \left( \sum_{z \in N_{as}(y)}\Mtld_G(z, x) \right) > \frac{1}{d_x}\sum_{z \in N_{as}(x)}\Mtld_G(z, y) \\
        \frac{1}{d_y}\sum_{z \in N_{as}(y)}\Mtld_G(z, x) &>& \frac{1}{d_x}\sum_{z \in N_{as}(x)}\Mtld_G(z, y).
    \end{eqnarray*}
    The last inequality contradicts the fact that $\Mtld_G(x,y)$ is the worst expected time, since by moving the agent with the state $y$ with probability $1$, we have a superior strategy than $\Mtld_G(x,y)$. Therefore we have that 
    $
        \Mtld_{S(x,y)}(x,y) \le 1 + \frac{1}{d_x} \sum_{z \in N_{as}(x)}\Mtld_G(z, y).
    $  
    Using the inequalities, the lemma holds. 
\end{proof}

\subsection{Proof of Lemma \ref{lma:averagephi}}
\label{apd:averagephi}
\averagephi*

\begin{proof}
    By the definition of $\tilde{\Phi}$ and proposition \ref{proposition:averagehit}, we have
    \begin{eqnarray*}
        \tilde{\Phi}(s_1, s_2) &=& \tilde{H}(s_1, \overline{s_2}) + \tilde{H}(s_2, \overline{s_{tu}}) - \tilde{H}(s_{tu}, \overline{s_2}) \\
        &=& 1 + \frac{1}{d_{s_1}} \left( \sum_{z \in N_{as}(s_1)} \tilde{H}(z, s_2) \right) + \tilde{H}(s_2, \overline{s_{tu}}) - \tilde{H}(s_{tu}, \overline{s_2}) \\
        &=& 1 + \frac{1}{d_{s_1}} \left( \sum_{z \in N_{as}(s_1)} \left( \tilde{H}(z, s_2) + \tilde{H}(s_2, \overline{s_{tu}}) - \tilde{H}(s_{tu}, \overline{s_2}) \right) \right) \\
        &=& 1 + \frac{1}{d_{s_1}} \sum_{z \in N_{as}(s_1)} \tilde{\Phi}(z, s_2).
    \end{eqnarray*}
    By the similar argument, $\tilde{\Phi}(s_2, s_1) = 1 + \frac{1}{d_{s_2}} \sum_{z \in N_{as}(s_2)} \tilde{\Phi}(z, s_1)$ holds. Therefore, by the equation $\tilde{\Phi}(s_1, s_2) = \tilde{\Phi}(s_2, s_1)$, the lemma holds. 
\end{proof}

\subsection{Proof of Theorem \ref{thrm:upperbound}}
\label{apd:upperbound}
\upperbound*

\begin{proof}[of Theorem \ref{thrm:upperbound}]
    For $x, y \in V$ and $s_{tu} \in \Sint$, we can write the potential function as
    \begin{eqnarray*}
        \tilde{\Phi}(x, y) &=& \tilde{H}(x, \overline{y}) + \tilde{H}(y, \overline{s_{tu}}) - \tilde{H}(s_{tu}, \overline{y}) \\
            &=& \tilde{H}(x, y) + \tilde{H}(y, s_{ut}) - \tilde{H}(s_{tu}, y)\\
            &=& \tilde{H}(x, y) + \tilde{H}(y, u) + \tilde{H}(u, s_{ut}) - \tilde{H}(u, y) - 1 \\
            &=& \tilde{H}(x, y) + \tilde{H}(y, u) - \tilde{H}(u, y) + 2(d_u - 1) + \sum_{z \in N_G(u) \setminus \{t\}} \tilde{H}(z, u) \\
            &=& 2 \left( H(x, y) + H(y, u) - H(u, y) + d_u - 1 + \sum_{z \in N_G(u) \setminus \{t\}} H(z, u) \right).
    \end{eqnarray*}
    In the fourth equality, we use Lemma \ref{lma:neighborhit}, and in the last equality we use the fact that $\tilde{H}(x, y) = 2 H(x, y)$ for $x, y \in V$. 
\end{proof}

\end{document}